\documentclass[final]{IEEEtran}

\usepackage{pgfplots}
\pgfplotsset{compat=newest}
\usepackage{epsfig}
\usepackage{amsthm,amssymb,graphicx,graphicx,multirow,amsmath,color}

\usetikzlibrary{plotmarks}
\usetikzlibrary{patterns}
\usetikzlibrary{calc,positioning,fit,backgrounds}
\usetikzlibrary{shapes,snakes}
\usetikzlibrary{intersections,positioning}
\usepackage{caption}
\usetikzlibrary{plotmarks}

\def\E{\mathbb{E}}

\def\P{\mathbb{P}}

\def\ie{{\em i.e.}}

\def\R{\mathbb{R}}

\def\one{\mathbf{1}}

\def\d{\mathrm{d}}

\def\I{\mathbf{I}}

\def\sinr{\mathtt{SINR}}

\newtheorem{lemma}{Lemma}{}
  \newtheorem{thm}{Theorem}
  \newtheorem{theorem}{Theorem}
  
  \newtheorem{cor}[thm]{Corollary}

\renewcommand{\baselinestretch}{0.95}

\title{Coverage Analysis in
	Millimeter Wave Cellular Networks with Reflections}
	
\author{\IEEEauthorblockN {Aroon Narayanan,  Sreejith T. V and  Radha Krishna Ganti} \\
\IEEEauthorblockA{Department of Electrical Engineering\\
Indian Institute of Technology, Madras\\
Chennai, India 600036\\
\{ee12b077, ee16ipf01, rganti\}@ee.iitm.ac.in}
}

\begin{document}
	\maketitle
	
\begin{abstract}	
The coverage probability of a user in a mmwave system  depends on the availability of line-of-sight paths or  reflected paths from any base station. Many prior works modelled blockages using random shape theory and analyzed the SIR distribution with and without interference.  While, it is intuitive that  the reflected paths do not significantly contribute to the coverage (because of longer path lengths), there are no works which provide a model and study the coverage with reflections. In this paper, we  model and analyze  the impact of reflectors using stochastic geometry. We observe that the reflectors have very little impact on the coverage probability.  

\end{abstract}

\section{Introduction}

Current cellular systems predominantly operate in the 1-6 GHz range of spectrum. In these frequencies, radio signals can propagate around an object, and it supports radio communication when a mobile device is blocked or shadowed by an obstruction. The next generation of wireless standards are looking at higher operating frequencies, mainly due to spectrum availability. Millimeter wave (mmWave) spectrum is the range of frequencies from 28-90 GHz, and is being envisioned  to augment the existing frequencies in the 5G standard \cite{pi2011}. Measurements have reaffirmed the feasibility of mmWave in the urban environment \cite{rap2013} and measurements for indoor communication at mmWave frequencies show that it holds promise with indoor stations \cite{rap2016}. Diffraction is a powerful propagation mechanism in today's 3G and 4G cellular systems but becomes very lossy at mmWave frequencies due to the small wavelengths of these bands. However, scattering and reflection become dominant at mmWave frequencies \cite{pi2011}.  Also, mmWave communication has been shown to be more sensitive to propagation loss than current modes of communication \cite{al2008}.

As has been shown in earlier works \cite{rap2002, bend2011}, first order reflections, \ie, paths from one point to another using one reflector, and second order reflections, \ie, paths from one point to another using two reflector, are important features at millimeter wave frequencies, especially by metallic objects. A later work \cite{raj2012} finds that well-known lossy objects such as human body and concrete are good reflectors at mmWave frequencies, enabling the receiver to capture secondary reflections for non-line-of-sight communication. Many other common objects have been shown in these works as having high reflection coefficients, which makes them a useful component of signal processing. Measurements for mmWave have revealed that the path loss characteristics for LOS and NLOS links are considerably different \cite{rap2013} \cite{rap20132} \cite{raj2012}.

The coverage of mmwave systems with blockages is analyzed in \cite{ak2014} using statistical models and in \cite{bai2012} \cite{bai2014} using tools from stochastic geometry. However, in these works, reflectors are not considered and only blockages are taken into account. However, as shown in  \cite{rap2002}, reflections   can contribute to the signal power, particularly if the LOS path is blocked. A later work \cite{bai2015} incorporated NLOS communication as well, but this model did not incorporate tractable randomly-placed reflectors.

In this paper, we look at the coverage in a mmwave system with both blockages and reflectors. Similar to the blockage model, we introduce a stochastic model for reflectors and analyse the coverage (SINR distribution). 
%
%

Section \ref{sec:Sysmodel} of the paper characterizes the system model. Section \ref{sec:Dist} discusses how the distance distributions of the reflected and direct paths can be derived. Section \ref{sec:Cov} contains the derivation of the coverage probability. Section \ref{sec:Res} discusses the results of the paper.
\section{System Model}\label{sec:Sysmodel}
We consider a wireless network on the 2-D plane in the
presence of both line-of-sight signal blockages and reflectors. Our main motivation is to characterise the effects of reflectors on the coverage probability of  a typical user.
\subsection{Base Stations}
The locations of the mmwave base stations are modelled by a spatial Poisson Point Process (PPP), $\Phi \subset \R^2$, with density $\lambda$.  A standard path loss model $l(x)=\|x\|^{-\alpha}$, $\alpha>2$, is assumed.  For any pair of nodes $x$ and $y$, independent Rayleigh fading (power) with unit mean is assumed  and is denoted by $h_{xy}$. 
The noise term is assumed to be circularly symmetric complex Gaussian noise (AWGN) with zero mean and variance $\sigma^2$.

\subsection{Blockages and Reflectors}
Random shape theory is used in this paper to populate the environment with objects. There are two types of objects - blockages and reflectors - and both of these are taken as 2D straight line segments with random  length and orientation that we describe later. Blockages are objects that are obstruction to any path, LOS or NLOS, between a base station and a user, but these do not reflect signals \ie, their reflection coefficient is zero. Reflectors are objects that can reflect the signals \ie, their reflection coefficient is non-zero. All objects in this model, both blockages and reflectors, have zero transmission coefficient, which means that all of them will attenuate to zero any signal that attempts to {\em pass through them}. This model of blockages has been used in \cite{bai2012}.

The centres of the objects form a spatial PPP $\Phi_o$ of density $\lambda_o$ of which a fraction $\delta$  are able to reflect the signals, \ie, the density of reflecting objects, $\lambda_R=\delta\lambda_o$ and that of blockages is $\lambda_R=(1-\delta)\lambda_o$. Hence the centres of the reflectors form a spatial PPP $\Phi_R$ of density $\lambda_R$ which is a thinned version of $\Phi_o$ and the centres of the blockages form a spatial PPP $\Phi_B$ of density $\lambda_B$.  Both the reflectors and blockages are assumed to be lines segments with random length $l$ and orientated at angle $\theta$ with the radial line from user to centres of the objects as shown in Fig.\ref{fig:deformPPP}. The dimension  $l$ is uniformly distributed in $[L_1, L_2]$, and the orientation of line segments $\theta$ is uniformly distributed in $[0, 2\pi)$. 

\subsection{Association Strategy}

A user can connect to a BS either by a direct path or a reflected path. We assume the user always connected to the BS which is having the shortest distance either through the direct path which is visible or through indirect path provided by the reflectors. Hence, the received signal-to-interference plus noise ratio ($\sinr$) is given by,

\begin{align}
	\one(r_d<r_r) \frac{P_t |h_d|^2r_d^{-\alpha}}{\sigma^2+I_D + I_R}+\one(r_r<r_d)\frac{P_t |h_r|^2r_r^{-\alpha}}{\sigma^2+I_D + I_R},
\end{align}
where  $r_d$ is the distance to the nearest visible BS and $r_r$ is the length of the shortest reflected path from any BS  through a reflector. Here $\one(.)$ is the indicator function and  $h_d$ and $h_r$ are the fading coefficients. $I_R$ and $I_D$ are the interference due to the reflected paths and the direct paths, respectively. 

\section{Distance Distributions}\label{sec:Dist}
In this section, we will describe the distribution of  $r_r$ and $r_d$ distances which characterize the $\sinr$.
\subsection{Distribution of shortest direct path}
When the BS density is $\lambda$ and that of blockages is $\lambda_B$, the distribution of distance to the closest visible base station, $r_d$, is derived in \cite{bai2012} and is given by,
\[
f_{R_d}(r_d) = 2\pi \lambda r_d\exp(-\lambda\pi r_d^2 S(r_d)-\beta r_d), 
\] where $\beta=2(\lambda_B+\lambda_R) L_b/\pi$  $= 2\lambda_o L_b/\pi$. Here  $L_b = \E[l] $ is the expected length of the blockage  and $S(x)=\frac{2}{\beta^2 x^2}(1-e^{-x \beta}(1+\beta x))$.
They have also shown that the probability of blockage depends on the length of the path. For a BS $x$, the probability that its path to the origin is blocked is given by
\begin{align}
	P_b(x) = \exp(-\beta \|x\|).
	\label{eq:block}
\end{align}
 
\subsection{Distribution of distance to nearest visible reflector}
We assume that a reflector is blocked if its centre is not visible from the origin (location of a typical UE). This is a reasonable assumption if the density of the reflectors and the length of the reflectors  is small.  Since the reflectors are distributed according to PPP of density $\lambda_R$ and blockage density $\lambda_B$, the distance to the centre of the closest visible reflector, $d$, is distributed according to:
\[
f_{D}(d) = 2\pi \lambda_R d\exp(-\lambda_R\pi d^2 S(d)-\beta d).
\]  
\subsection{Distribution of shortest reflected path}
\begin{figure}
	\begin{center}
		\begin{tikzpicture}[scale=0.6]
		\draw[fill=gray!20,draw=black!0] (1,3)--(3,3) -- (6,6) -- (2,6)-- cycle;
		\draw[fill=blue!20,draw=black!0] (1,3)--(3,3) -- (8,-1) --(3,-1) -- cycle;
		\coordinate (x) at (0,0);
		\fill (0,0) circle [radius=4pt];
		\node [draw,scale=0.4,regular polygon,regular polygon sides=3,fill=black!100]at (5.7,5) {}; 
		\draw (6.6,5) node  {$x\in\Phi$};
		\node [draw,scale=0.4,regular polygon,regular polygon sides=3,fill=red!80]at (5.7,4.2) {}; 
		\draw (6.6,4.2) node  {$x^\prime\in\Phi^\prime$};
		\fill (5.7,3.5) circle [radius=4pt];
		\draw (6.6,3.5) node  {UE};
		\node [fill,draw,scale=0.3,circle]at (x) {};  
		\draw(-0.5,0.4) node {$O$};
		
		\node [draw,scale=0.3,regular polygon,regular polygon sides=3,fill=black!100]at (3.2,1) {}; 
		\draw(3.6,1) node {$p$};
		\draw(4,0) node {$A_1$};

		\node [draw,scale=0.4,regular polygon,regular polygon sides=3,fill=red!80]at (3,4.5) {};
		
		\draw(3.6,4.5) node {$p\prime$};
		\draw(4.6,5.5) node {$A_2$};
		
		\node [draw,scale=0.3,regular polygon,regular polygon sides=3,fill=black!100]at (4,1.4) {}; 
		\node [draw,scale=0.4,regular polygon,regular polygon sides=3,fill=red!80]at (3.7,4.2) {};
		\node [draw,scale=0.3,regular polygon,regular polygon sides=3,fill=black!100]at (3.2,-0.4) {}; 
		\node [draw,scale=0.4,regular polygon,regular polygon sides=3,fill=red!80]at (3.2,5.6) {};
		\node [draw,scale=0.3,regular polygon,regular polygon sides=3,fill=black!100]at (-0.5,4.4) {};
		\node [draw,scale=0.3,regular polygon,regular polygon sides=3,fill=black!100]at (2,-0.7) {};
		
		\draw (2,3) -- (0,0);
		\draw (2,3) -- (3.2,1);
		\draw (2,3) -- (3,4.5);
		
		\draw (1,3) -- (3,3);
		\draw [dashed] (1,3) -- (1,1.7);
		\draw [dashed] (3,3) -- (3,1.7);
		\draw (0,0) -- (1,3);
		\draw (1,3) -- (3,-1);
		\draw (0,0) -- (3,3);
		\draw (3,3) -- (8,-1);
		\draw (1,3) -- (2,6);
		\draw (3,3) -- (6,6);
		
		\draw(0.7,3.3) node {$A$};
		\draw(3,3.3) node {$B$};
		
		\end{tikzpicture}
	\end{center}
	\caption[]{Illustration of deformed PPP. Here, $AB$ is the reflector which reflects the signal from BS $p$ to the user located at $O$. Original PPP (black triangles) is $\Phi$ and the deformed PPP (red triangle) about $AB$ is denoted by $\Phi^\prime$. The reflected point in the transformed PPP is shown as $p^\prime$. }
	\label{fig:deformPPP}
\end{figure}
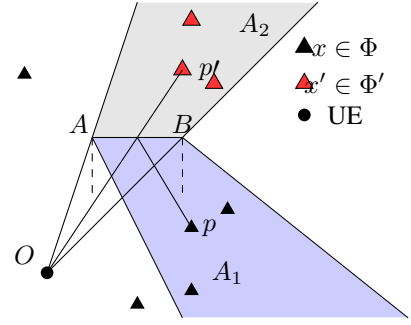

In  Fig. \ref{fig:deformPPP}, $AB$ represents a mirror and $O$ represents the UE receiver at the origin. We first observe that only the BS in the shaded  region  $A_1 \subset \R^2$ can reflect of $AB$ and reach the origin. Signals from BSs in $\R^2\setminus A_1$ cannot reach the origin through the reflector $AB$. The total length that a signal propagates for a reflected path equals the distance to the mirror and the distance from the mirror to the user at the origin.  In order to characterise the distribution of the shortest reflected path, we project each of the possible base stations as shown in Fig. \ref{fig:deformPPP}. The region $A_2$ denotes the mirror image of the region $A_1$ and for computing the length of the reflected path, a BS  $y \in A_1$ can be mirror imaged to $y\prime \in A_2$. The number of reflected base stations in the semi-conic region on the other side will be always equal to the number of base stations on the original side due to symmetry, so the density of base stations of the transformed PPP in the reflected semi-conic region will be $\lambda_B$ in the region $A_2$. These reflected points may or may not be visible to the user, depending on whether there are blockages on the link between the points and the user.  We now characterise the length of the shortest reflected path for a reflector at a distance $d$.

\begin{lemma}\label{lem:rfconD}
	Given $\lambda$, $\lambda_B$ and $\lambda_R$ are the densities of BSs, blockages and reflectors respectively, the conditional CDF of the distance of shorted reflected path, $r_r$, when the nearest reflector is at a distance $d$ is given by,
	\begin{align*} 
	\P(R_f>r_r | d, \theta)= \E_l\exp\left(- \int_{\mathcal{A}(\theta,d)}\lambda e^{-\beta ||x||}\d x\right),
	\end{align*} 
	where $\mathcal{A}(\theta,d)=B(0,r_r)\cap A_2(\theta,d)$. Here $\theta$ represents the orientation of the reflector.
\end{lemma}
 \begin{proof}
	Suppose that the user is located at the origin, $O$ and the nearest reflector is at a distance $d$ from it as shown in  Fig.\ref{fig:refldist}. Let the shortest visible reflecting BS after the transformation of $\phi$, $p^\prime\in\phi^\prime$, is at a distance $x$ from the user at origin. Remember $x$ is the total distance of the reflected path in the original PPP. Note that the distance to the nearest visible BS $p^\prime$ is larger than $x$ if and only if all the base stations located within the shaded region $\mathcal{A}(\theta,d)$ are not visible to the user. Let  $V_{x|d}$ denote the event that $x$ is visible given $d$. 	We have 
	\begin{align*}
	\P(R_f>r_r|d)
	&=\E\left[\prod_{x\in \mathcal{A}(\theta,d)\cap \Phi}\left(1-\one(V_{x|d})\right))\right],\\
	&\stackrel{(a)}{=}\E\left[ \prod_{x\subset \mathcal{A}(\theta,d)\cap \Phi}(1-e^{-\beta||x||})\right],\\
	&\stackrel{(b)}{=}\E_l\exp\left(- \int_{\mathcal{A}(\theta,d)}\lambda e^{-\beta ||x||}\d \textbf{x}\right),
	\end{align*}
	where $(a)$ follows from \eqref{eq:block} and $(b)$ follows from the PGFL property of a PPP. 
\end{proof}
Observe that  $\mathcal{A}(\theta,d)=B(0,r_r)\cap A_2(\theta,d)$ is a complicated region for integration. To simplify the integral further, we make the following assumption: {\em We assume that the reflectors are perpendicular to the line connecting the UE and the reflector's center,  \ie, $\theta =\pi/2$.} 

The distribution of the shortest reflecting path length is given in the following corollary. 

\begin{cor}\label{lem:rfconD}
	Given $\lambda$, $\lambda_B$ and $\lambda_R$ are the densities of BSs, blockages and reflectors respectively, the conditional distribution of the distance of shorted reflected path, $r_r$, when the nearest reflector is at a distance $d$ is given by,
	\begin{align*} 
	f_{R_f}(r_r|d)\approx \E_l\left[ \lambda\theta_d K^{'}(r_r)e^{-\lambda\theta_d K(r_r)}\right],
	\end{align*} 
	where $\theta_d=\arctan(l/2d)$ and $K(r_r)$ is given by
	\[ 
	K(r_r) 
	= 2\left[\frac{e^{-\beta d}}{\beta}(d+\frac{1}{\beta})-\frac{e^{-\beta r_r}}{\beta}(r_r+\frac{1}{\beta})\right]
	\] and $K^{'}(r_r)=2  \left(e^{-\beta r_r} \left(\frac{1}{\beta }+r_r\right)-\frac{e^{-\beta r_r}}{\beta }\right)$ is the first derivative of $K(r_r) $. Here $\E_l$ is the expectation with respect to length of the reflector.
\end{cor}

\begin{proof}
	\begin{align*}
	\P(R_f>r_r|d)
	&=\exp\left(- \int_{\mathcal{A}}\lambda e^{-\beta ||x||}\d \textbf{x}\right)\\
	&\stackrel{(a)}{\approx}\exp\left(- \int_{\theta = -\theta_d}^{\theta_d} \int_{r=d}^{r_r} \lambda e^{-\beta r}r \d r \d \theta \right)\\
	&=\exp(-\lambda\theta_d K(r_r)).
	\end{align*}
	Here in (a), area of shaded region $\mathcal{A}\approx \theta_d (r_r^2-d^2)$ by approximating the reflector perpendicular to $d$ as an arc of length $l$ which subtend $\angle COB=\theta_d$. Then by finding negative derivative, $-\frac{\d}{\d r}\P(R_f>r_r|d)$ gives the PDF.
\end{proof}
\begin{figure}
	\begin{center}
	\includegraphics[scale=0.08]{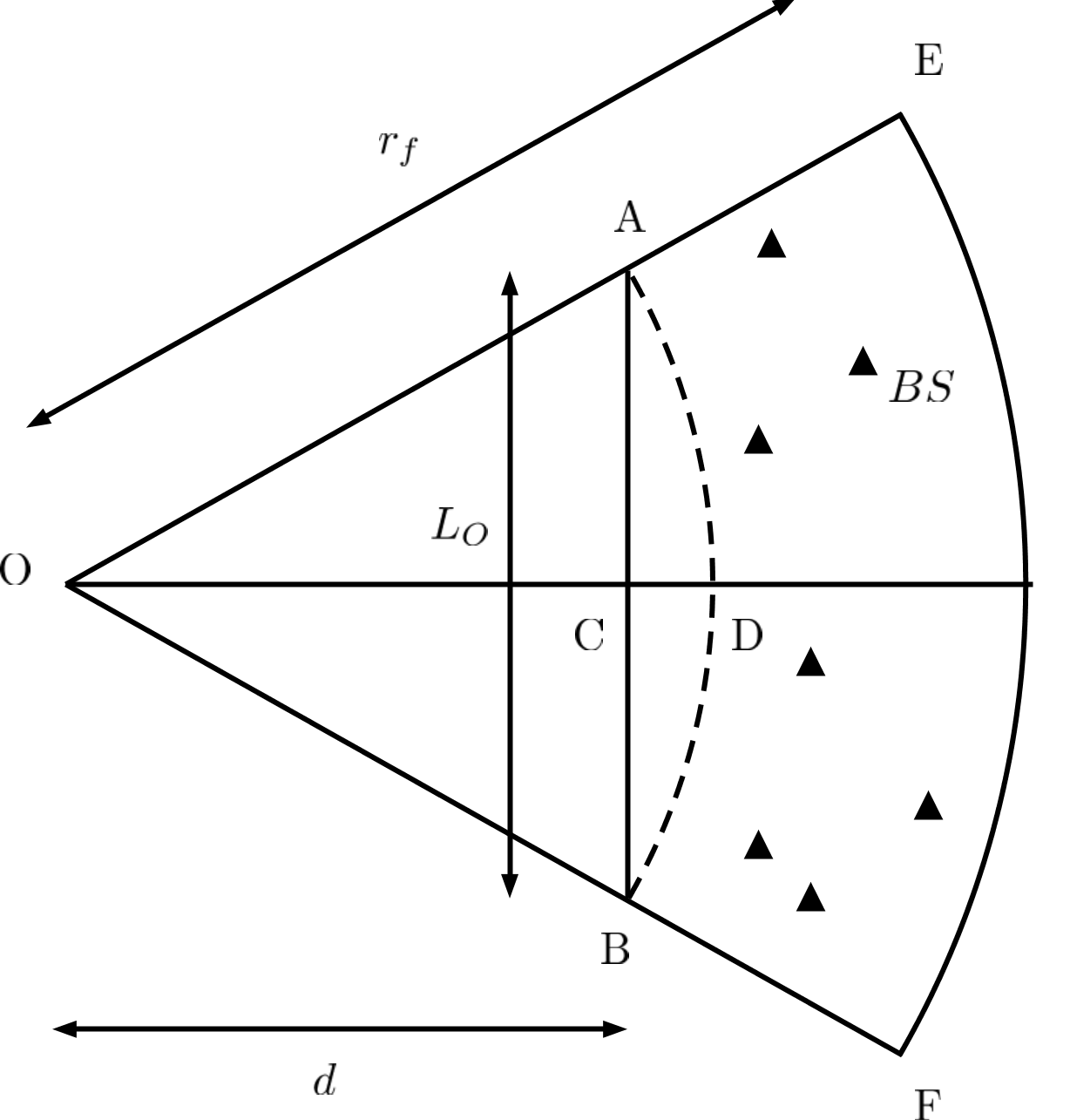}
	\end{center}
	\caption[]{Illustration of approximation of shortest reflection path distribution. $AB$ is the reflector of length $l$ and at a distance $OC=d$ from the user at $O$. Also $AB$ is assumed perpendicular to the radial line from user to reflector centre. The region $\mathcal{A}$ over which integration is done is approximated to the region between the two arcs, labelled $ABFE$. The nearest reflected BS is $p^\prime$ which is at a distance $r_r$. }
	\label{fig:refldist}
\end{figure}

Assuming that that the shortest reflected path is through the closest visible reflector, the distribution of $r_r$ can be  obtained  by unconditioning the conditional density function using the following property of conditional and joint distributions functions, $f(x,y)=f(x|y)f(y)$. 
Consider the CCDFs of the direct path length and reflected path length, it can be easily verified that  \[\P[R_d=\infty]=e^{-\frac{2\pi\lambda}{\beta^2}},\] 
which implies that there is a finite probability that there is no direct path to the typical user at the origin. This might happen when all the BSs to the user  at the origin are blocked. Similarly,  
\begin{align*}
\P[R_r=\infty| d]&=\E_le^{- \int_{\mathcal{A}_2(\theta,d)}\lambda e^{-\beta ||x||}\d \textbf{x}}\\
 &\approx \E_l\left[ \exp(-2\lambda \theta_d e^{-\beta d}(d/\beta + 1/\beta^2)  ) \right],
\end{align*}
 which is finite and non-zero for all combinations of  $\lambda,\lambda_o$. This implies that there is no reflected path. This result is different from the the conventional cellular case, $\P[R>r]=e^{-\lambda r^2}$ which tends to zero as $r\to \infty$. The presence of blockages not just increases the shortest connected path but can also make a user to be uncovered by the network. Mean distances  are plotted in Figure \ref{fig:meandist}. From Figure \ref{fig:meandist}, we can see that both  the direct distance and the reflected distances are decreasing with increasing BS density.   We can also see that when the dimension of reflectors increases, the shortest reflected path distance decreases.  That makes sense as as it is more likely for a user to get more reflections  with increasing length of reflectors.  
\begin{figure}[h]
	\centering 
	\begin{tikzpicture}
	\begin{axis}[scale=0.95,
	grid = both,
	legend style={
		cells={anchor=west},
		legend pos= north east ,
	},
	xlabel ={$\log(\lambda)$},
	ylabel = {Mean distances}
	]
	\addplot[mark=square*, mark options={fill=black},thick, mark size=2.5]
	coordinates{
		(-4.000000, 54.247700) (-3.500000, 29.379100) (-3.000000, 16.197800) (-2.500000, 9.011470) (-2.000000, 5.037610) (-1.500000, 2.823540) (-1.000000, 1.584870) (-0.500000, 0.890317) (0.000000, 0.500372) (0.500000, 0.281288) (1.000000, 0.158151) (1.500000, 0.088926) (2.000000, 0.050004) 
	};\addlegendentry{$U(1,10),\delta=0.2$ D}
	
	\addplot[mark=triangle*, mark options={fill=white},thick,  mark size=3.5]
	coordinates{
		(-4.000000, 272.288000) (-3.500000, 190.566000) (-3.000000, 107.792000) (-2.500000, 66.077800) (-2.000000, 48.555000) (-1.500000, 41.452900) (-1.000000, 38.779200) (-0.500000, 37.846100) (0.000000, 37.537000) (0.500000, 37.437400) (1.000000, 37.405600) (1.500000, 37.395600) (2.000000, 37.392400) 
	};\addlegendentry{$U(1,10),\delta=0.2$ R}

	\addplot[mark=square*, mark options={fill=white},thick,red, mark size=2.5]
	coordinates{
		(-4.000000, 61.096600) (-3.500000, 39.086400) (-3.000000, 18.428700) (-2.500000, 9.620790) (-2.000000, 5.217200) (-1.500000, 2.878270) (-1.000000, 1.601840) (-0.500000, 0.895624) (0.000000, 0.502040) (0.500000, 0.281814) (1.000000, 0.158317) (1.500000, 0.088978) (2.000000, 0.050020) 
	};\addlegendentry{$U(5,50),\delta=0.5$ D}

	\addplot[mark=triangle*, mark options={fill=red},thick,red,  mark size=3.5]
	coordinates{
		(-4.000000, 243.604000) (-3.500000, 145.793000) (-3.000000, 78.265100) (-2.500000, 46.840000) (-2.000000, 32.844700) (-1.500000, 26.818500) (-1.000000, 24.437600) (-0.500000, 23.576900) (0.000000, 23.285800) (0.500000, 23.191100) (1.000000, 23.160800) (1.500000, 23.151200) (2.000000, 23.148100) 
	};\addlegendentry{$U(1,10),\delta=0.5$ R}

	\addplot[mark=halfcircle*, mark options={fill=blue},thick,blue, mark size=3]
	coordinates{
		(-4.000000, 29.293900) (-3.500000, 51.385000) (-3.000000, 61.345200) (-2.500000, 58.887100) (-2.000000, 54.690100) (-1.500000, 52.613200) (-1.000000, 51.897700) (-0.500000, 51.440000) (0.000000, 51.449700) (0.500000, 51.308800) (1.000000, 51.659800) (1.500000, 51.680400) (2.000000, 51.687900) 
	};\addlegendentry{$U(10,50),\delta=0.2$ R}
	
	\addplot[mark=diamond*,thick,mark options={fill=black},black, mark size=3.5]
	coordinates{
		(-4.000000, 40.665700) (-3.500000, 59.679400) (-3.000000, 53.313500) (-2.500000, 39.646800) (-2.000000, 32.440000) (-1.500000, 29.803500) (-1.000000, 28.918900) (-0.500000, 28.629900) (0.000000, 28.535200) (0.500000, 28.503800) (1.000000, 28.495700) (1.500000, 28.493200) (2.000000, 28.495000) 
	};\addlegendentry{$U(10,50),\delta=0.5$ R}

	\end{axis}
	\end{tikzpicture}
	\caption{Mean of shortest direct path and reflected path lengths Vs BS density for  $\lambda_B=10^{-3}$, $\lambda_R=10^{-3}$, \ie, $\lambda_o=2\times 10^{-3}$ with different $\delta$ and different dimensions for objects. The dimension of objects are distributed uniformly, $l\sim U(L_1,L_2)$.}
	\label{fig:meandist}
\end{figure}
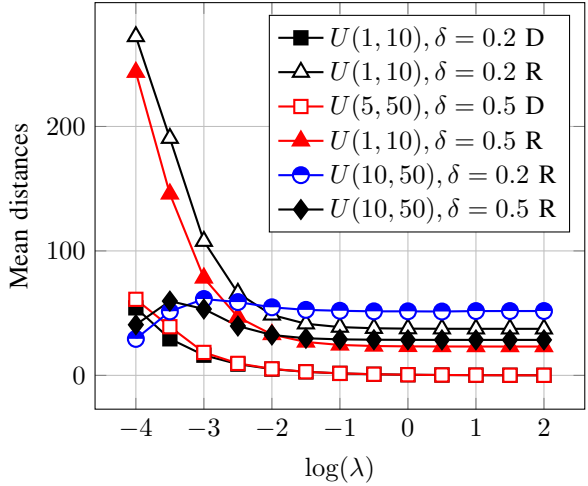

\subsection{Association Probabilities}
Since the association is assumed to be nearest connectivity either through direct visible path or a reflected path, the probability of direct connectivity for a typical user is,
\begin{align}
\label{eq:PrD}
p_d=\P(r_d>r_r)=\int_{0}^{\infty}\int_{r_r}^{\infty} f_{R_f}(r_r)f_{R_d}(r_d)\d r_d \d r_r
\end{align}
Similarly, the probability of reflected connectivity is,
\begin{align}
\label{eq:PrR}
p_r=\P(r_d<r_r)=\int_{0}^{\infty}\int_{r_d}^{\infty} f_{R_f}(r_r)f_{R_d}(r_d) \d r_r \d r_d
\end{align}

\section{Coverage Probability}\label{sec:Cov}

Coverage probability for a user is defined as the probability that the signal received by the user has an $\sinr$ greater than a threshold, $T$ to establish the connectivity. In theorem.\ref{th:main} we will provide the coverage probability for a typical user.  
\begin{theorem}\label{th:main}
	The coverage probability $\P_C(T)$ for a user connected to base stations either through the direct path or the reflected path is 
	\begin{align}
	\label{eq:Main}
	\P_C(T) = P_{D}(T) + P_{R}(T),
	\end{align}	
where $P_{D}(T)$ is the coverage probability of a user connected through the direct path, given by,

\begin{align*}P_{D}(T) &= \E_{r_d<r_r}\Big[e^{-2\pi\lambda\int_{r_d}^{\infty}\left( \frac{Tr_d^\alpha r^{-\alpha}e^{-\beta r}}{1+Tr_d^\alpha r^{-\alpha}}\right)r\d r}\Big.\nonumber\\
&\Big.\times \left(1-\frac{Tr_d^\alpha r_r^{-\alpha}e^{-\beta r_r}}{1+Tr_d^\alpha r_r^{-\alpha}}\right) e^{-r_d^{\alpha}{\sigma^2}T} \Big].
\end{align*}
$P_{R}(T)$ is the coverage probability of a user connected through the reflected path, given by,
\[
P_{R}(T) = \E_{r_d>r_r}\Big[e^{-2\pi\lambda\int_{r_r}^{\infty}\left(\frac{Tr_r^\alpha r^{-\alpha}e^{-\beta r}}{1+Tr_r^\alpha r^{-\alpha}}\right)r\d r} e^{-r_r^{\alpha}{\sigma^2}T}\Big].
\]
 \end{theorem}
\begin{proof}
See Appendix.
\end{proof}

%
%

\section{Results and Discussion}\label{sec:Res}

In this Section, we numerically evaluate the coverage probability given in Theorem \ref{th:main} and compare with our theoretical derivations. We simulated a square area in which base stations, blockages and reflectors are distributed according to  PPPs such that there are at least $100$ base stations on average in the area  and $\alpha =4$. The objects (blockages and reflectors) had lengths chosen from a uniform distribution, $U(L_1,L_2)$ and the orientation of reflectors and blockages are uniformly distributed \ie, $\theta\sim U(0,2\pi)$.  It must be noted that our approximations in our analysis is valid only for cases in which the length of objects is comparable to, or smaller than, the mean distance between objects, \ie, $\min(\frac{1}{2\sqrt{\lambda_O}},\frac{1}{2\sqrt{\lambda}})$. This is because in analysis we approximated reflector as an arc, however this makes the objects not overlapping  as in the practical scenario.

The average lengths of shortest visible direct path and  the reflected path through nearest visible reflector are given in Table \ref{tab:dist_meandelta}. We observe that as the reflector density is increased, the reflected path length shortens as there are more reflectors. The shortest direct path length is not varying as the density of blockages remain the same. We also observe that undercertain configurations, the average length of the reflected path is very similar to the length of the direct path. 
\begin{table}[!ht]
	\renewcommand{\arraystretch}{1.2}
	\centering \caption{Mean distances of shortest direct path and reflected path for different fractions of reflectors. $r_c$ denotes the connected BS distance in no blockage and no reflectors case.} 
	\label{my-label}
	\begin{tabular}{|c|c|c|c|c|c|c|}
		\hline
		\multirow{2}{*}{$(\lambda, \lambda_o)$}&\multirow{2}{*}{$l$} &\multirow{2}{*}{$r_c$} &\multicolumn{2}{c|}{$r_d$} & \multicolumn{2}{c|}{$r_r$} \\ \cline{4-7} 
		&	& & $\delta=0.2$       & 0.5       & 0.2       & 0.5       \\ \hline
		$(10^{-1},10^{-1})$& (1,5)	& 1.58 & 1.84  & 1.84 &  6.03 &  5.94  \\ \hline
		$(10^{-1},10^{-2})$&(1,5) & 1.58	& 1.61    & 1.61  & 16.23  & 10.70  \\ \hline
		$(10^{-2},10^{-2})$&(1,10) & 5 & 5.42 & 5.42 & 29.49 & 23.69 \\ \hline
		$(10^{-3},10^{-3})$& (1,10) & 15.81 & 16.19 & 16.19 & 151.94 & 112.29 \\ \hline
	\end{tabular}
	\label{tab:dist_meandelta}
\end{table}

From the simulations and theory, the probability of the user connecting to the direct path is very high for most cases, so mostly the user will be tagged to nearest visible BS instead of getting connected through a reflected path. These probabilities are plotted in Figure \ref{fig:probcon}.
\begin{figure}[h]
	\centering 
	\begin{tikzpicture}
	\begin{axis}[scale=1,
	grid = both,
	legend style={at={(0.98,0.75)}},
	xlabel ={$\delta$},
	ylabel = {$p_d,p_r$}
	]
	
	
	\addplot[mark=square*, mark options={fill=red},thick, mark size=2.5]
	coordinates{
		(0.200000, 0.97743) (0.500000, 0.950739) (0.800000, 0.933093) 
	};\addlegendentry{$(10^{-3},10^{-3})$}
	
	\addplot[mark=halfcircle*, mark options={fill=blue},thick, mark size=3]
	coordinates{
		(0.200000, 0.997619) (0.500000, 0.994164) (0.800000, 0.990827) 
	};\addlegendentry{$(10^{-3},10^{-4})$}
	
	\addplot[mark=triangle*, mark options={fill=blue},thick,  mark size=3]
	coordinates{
		(0.200000, 0.992) (0.500000, 0.981) (0.800000, 0.97) 
	};\addlegendentry{$(10^{-1},10^{-2})$}
	
	\addplot[mark=diamond*, mark options={fill=white},thick, mark size=3]
	coordinates{
		(0.200000, 0.905) (0.500000, 0.897) (0.800000, 0.869) 
	};\addlegendentry{$(10^{-2},10^{-2})$}
	
	\addplot[mark=square*,dashed, mark options={fill=blue},thick, mark size=2.5]
	coordinates{
		(0.200000, 0.106121) (0.500000, 0.168386) (0.800000, 0.202862) 
	};\addlegendentry{$(10^{-3},10^{-2})$}
	
	\addplot[mark=square*,dashed, mark options={fill=red},thick, mark size=2.5]
	coordinates{
		(0.200000, 0.021239) (0.500000, 0.0460507) (0.800000, 0.0654729) 
	};
	
	\addplot[mark=halfcircle*, mark options={fill=blue},thick,only marks, mark size=3]
	coordinates{
		(0.200000, 0.00237741) (0.500000, 0.00583564) (0.800000, 0.00917199) 
	};
	
	\addplot[mark=triangle*,dashed, mark options={fill=blue}, thick, mark size=3]
	coordinates{
		(0.200000, 0.007) (0.500000, 0.018) (0.800000, 0.028) 
	};
	
	\addplot[mark=diamond*,dashed, mark options={fill=white},thick, mark size=3]
	coordinates{
		(0.200000, 0.042) (0.500000, 0.0888) (0.800000, 0.1237) 
	};
	
	\end{axis}
	\end{tikzpicture}
	\caption{Probability of connecting through shortest direct path and shortest reflected path Vs different relative reflector density factor, $\delta$. The BS density and objects densities are given in as $(\lambda,\lambda_o)$, length of objects $l\sim U(L_1,L_2)$m, solid lines represent the direct connection and dashed lines for reflected connection.}
	\label{fig:probcon}
\end{figure}
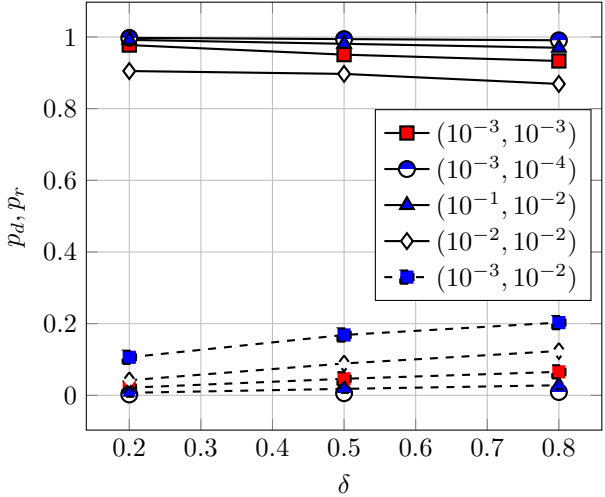

As expected, the probability of the user connecting to the reflected path increases  in urban environments with a high density of reflecting objects, such as metallic objects, as compared to the density of base stations. This can be seen from  $(10^{-3}, 10^{-2}),\delta=0.2$ and $(10^{-3}, 10^{-2}),\delta=0.5$. This implies that a measurable gain in coverage is possible by reflectors.  This is can be verified from coverage probability Fig.\ref{fig:CovCompCBR2}.
Also, our theory does not consider reflections from the reflectors other than the closest one. When the density of reflector is high, it is likely that the user can connect to a BS through a shortest reflected path from reflector other than the closest reflector as well. 


In Figures \ref{fig:CovCompCBR} and \ref{fig:CovCompCBR2}, the coverage probability is plotted for various scenarios. We first observe that the MonteCarlo simulations match the results. We also observe that the improvement in coverage is very minor because of the reflectors.  For reference, we have also  plotted the coverage probability of  a  networks without any blockages and reflectors.

\begin{figure}[ht]
	\centering
	\begin{tikzpicture}
	\begin{semilogyaxis}[scale=1,
	grid = both,
	legend style={
		cells={anchor=west},
		legend pos=south west ,
	},
	legend style={font=\tiny},
	xlabel = $\sinr$ threshold $T$,
	ylabel = $\P(\sinr>T)$
	]
	
	\addplot[mark=square*, mark options={fill=red}, mark size=2.5]
	coordinates{
		(-5.000000, 0.776355) (-2.000000, 0.652226) (1.000000, 0.513960) (4.000000, 0.384993) (7.000000, 0.279634) (10.000000, 0.200050) (13.000000, 0.142193) (16.000000, 0.100814) (19.000000, 0.071409) 
	};\addlegendentry{No blockages/reflectors}
	
	\addplot[mark=triangle*, mark options={fill=blue},  mark size=3]
	coordinates{
		(-5.000000, 0.797718) (-2.000000, 0.679709) (1.000000, 0.543625) (4.000000, 0.412560) (7.000000, 0.302768) (10.000000, 0.218336) (13.000000, 0.156171) (16.000000, 0.111289) (19.000000, 0.079161) 
	};\addlegendentry{$\lambda=10^{-1}, \delta=0$}
	
	\addplot[mark=diamond*, mark options={fill=white}, mark size=3]
	coordinates{
		(-5.000000, 0.790680) (-2.000000, 0.670915) (1.000000, 0.534142) (4.000000, 0.403655) (7.000000, 0.295205) (10.000000, 0.212297) (13.000000, 0.151519) (16.000000, 0.107784) (19.000000, 0.076558) 
	};\addlegendentry{$\lambda=10^{-1}, \delta=0.2$}

	\addplot[mark=halfcircle*, mark options={fill=black}, mark size=3]
	coordinates{
		(-4.000000, 0.754026) (-1.000000, 0.625689) (2.000000, 0.488160) (5.000000, 0.363764) (8.000000, 0.263938) (11.000000, 0.189059) (14.000000, 0.134665) 
	};\addlegendentry{$\lambda=10^{-1}, \delta=0.5$}

	\addplot[mark=oplus, mark options={fill=black},thick,only marks, mark size=3]
	coordinates{
		(-5.000000, 0.811000) (-2.000000, 0.678000) (1.000000, 0.548000) (4.000000, 0.416000) (7.000000, 0.297000) (10.000000, 0.221000) (13.000000, 0.151000) (16.000000, 0.106000) (19.000000, 0.070000) 
	};\addlegendentry{Monte Carlo: $\lambda=10^{-1},  \delta=0.5$}
	
	\addplot[mark=*, mark options={fill=blue},  mark size=2.5]
	coordinates{
		(-5.000000, 0.612416) (-2.000000, 0.559091) (1.000000, 0.487280) (4.000000, 0.404908) (7.000000, 0.323247) (10.000000, 0.250678) (13.000000, 0.190635) (16.000000, 0.143097) (19.000000, 0.106482) 
	};\addlegendentry{$\lambda=10^{-2},\delta=0.5$}
	
	\addplot
	coordinates{
			(-5.000000, 0.070000) (-2.000000, 0.064000) (1.000000, 0.058000) (4.000000, 0.050000) (7.000000, 0.045000) (10.000000, 0.036000) (13.000000, 0.032000) (16.000000, 0.026000) (19.000000, 0.020000) 
	};\addlegendentry{$\lambda=10^{-3},\delta=0.5$}
	
	\end{semilogyaxis}
	\end{tikzpicture}
	\caption[]{Comparison of coverage probabilities for cellular case without blockages and reflections and with blockages only and with both blockages and reflections. The BS density $\lambda=10^{-1}$, total object density $\lambda_o=10^{-2}$ of which $\delta$ percentage of objects are reflecting. The length of objects are assumed to be uniformly distributed in all cases. }
	\label{fig:CovCompCBR}
\end{figure}
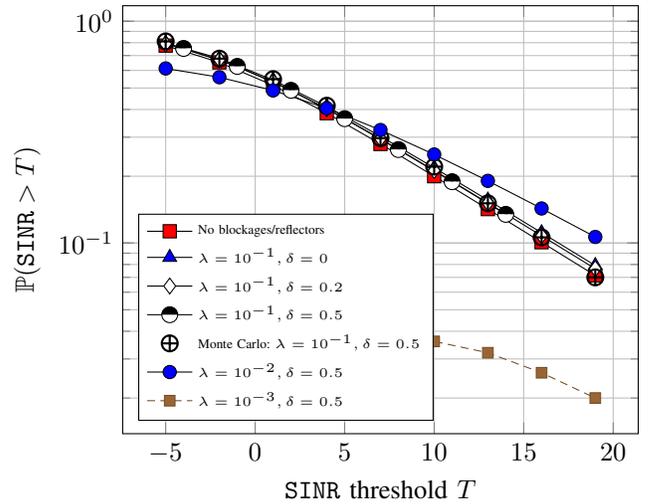

\begin{figure}[ht]
	\centering
	\begin{tikzpicture}
	\begin{semilogyaxis}[scale=1,
	grid = both,
	legend style={
		cells={anchor=west},
		legend pos=south west ,
	},
	legend style={font=\tiny},
	xlabel = $\sinr$ threshold $T$,
	ylabel = $\P(\sinr>T)$
	]
	
	\addplot[mark=square*, mark options={fill=red}, mark size=2.5]
	coordinates{
		(-5.000000, 0.776355) (-2.000000, 0.652226) (1.000000, 0.513960) (4.000000, 0.384993) (7.000000, 0.279634) (10.000000, 0.200050) (13.000000, 0.142193) (16.000000, 0.100814) (19.000000, 0.071409) 
	};\addlegendentry{No blockages/reflectors}
	
	\addplot[mark=triangle*, mark options={fill=blue},  mark size=3]
	coordinates{
		(-5.000000, 0.797718) (-2.000000, 0.679709) (1.000000, 0.543625) (4.000000, 0.412560) (7.000000, 0.302768) (10.000000, 0.218336) (13.000000, 0.156171) (16.000000, 0.111289) (19.000000, 0.079161) 
	};\addlegendentry{$\delta=0$}
	

	\addplot[mark=halfcircle*, mark options={fill=black}, mark size=2.5]
	coordinates{
		(-5.000000, 0.781265) (-2.000000, 0.662687) (1.000000, 0.527596) (4.000000, 0.398729) (7.000000, 0.291550) (10.000000, 0.209563) (13.000000, 0.149448) (16.000000, 0.106206) (19.000000, 0.075356) 
	};\addlegendentry{$\delta=0.5$}
	
	\addplot[mark =*,  mark size=2.5,thick]
	coordinates{
		(-5.000000, 0.775353) (-2.000000, 0.652035) (1.000000, 0.514338) (4.000000, 0.385612) (7.000000, 0.280269) (10.000000, 0.200603) (13.000000, 0.142639) (16.000000, 0.101159) (19.000000, 0.071670) 
	};\addlegendentry{$\lambda_o=10^{-4}, \delta=0.2$}
	
	\addplot[mark=square*, mark options={fill=white},thick, mark size=2.5]
	coordinates{
		(-5.000000, 0.260684) (-2.000000, 0.239560) (1.000000, 0.212225) (4.000000, 0.180930) (7.000000, 0.149221) (10.000000, 0.120040) (13.000000, 0.094887) (16.000000, 0.074108) (19.000000, 0.057422) 
	};\addlegendentry{$\lambda_o=10^{-2}, \delta=0.2$}
	
	\addplot[mark=oplus, mark options={fill=blue},thick, mark size=3]	
	coordinates{
		(-5.000000, 0.362997) (-2.000000, 0.333625) (1.000000, 0.295814) (4.000000, 0.252604) (7.000000, 0.208779) (10.000000, 0.168331) (13.000000, 0.133326) (16.000000, 0.104276) (19.000000, 0.080842) 
	};\addlegendentry{$\lambda_o=10^{-2}, \delta=0.5$}

	\addplot[mark=diamond*, mark options={fill=white}, mark size=3,only marks]
	coordinates{
		(-5.000000, 0.807000) (-2.000000, 0.683000) (1.000000, 0.533000) (4.000000, 0.404000) (7.000000, 0.287000) (10.000000, 0.202000) (13.000000, 0.144000) (16.000000, 0.108000) (19.000000, 0.074000) 
	};
	
%
	%
	\end{semilogyaxis}
	\end{tikzpicture}
	\caption[]{Comparison of coverage probabilities for cellular case without blockages and reflections and with blockages only and with both blockages and reflections. The BS density $\lambda=10^{-3}$, total object density $\lambda_o=10^{-3}$ of which $\delta$ percentage of objects are reflecting. The length of objects are assumed to be uniformly distributed in all cases. \ie, $l\sim U(1,10)$m.}
	\label{fig:CovCompCBR2}
\end{figure}
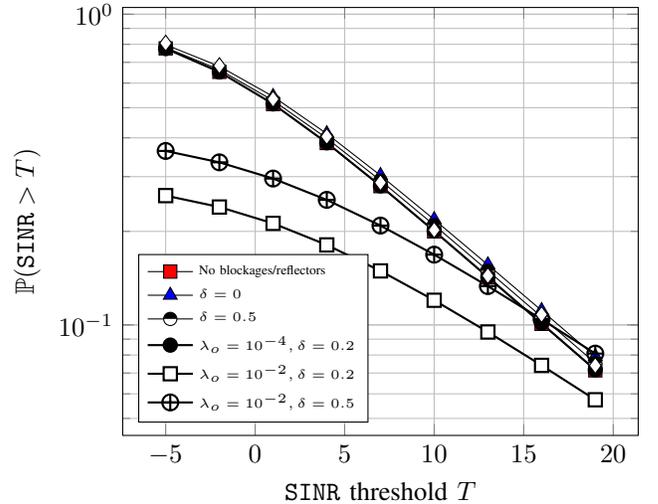
When the number of objects is increased while keeping the number of base stations fixed, there are multiple effects to be considered. There will be a higher probability of blockage due to higher number of objects, and added reflectors can contribute useful signal as well as interference. It is well established result that, the coverage probability of cellular network in interference limited scenario is independent of the BS density \cite{ganti}. It can be seen that the presence of objects changes the coverage probability with density of BS. For high dense network with moderate density of blockages and reflectors improve the coverage probability, intuitively we can say that blockages reduces the interfering signals. For low density networks, presence of high density of reflectors and blockages reduces the coverage probability drastically as most of the links are blocked by the objects. For dense networks when the blockages and reflectors are introduced, the coverage probability improves, because of the eliminating a large portion of interference at the same time direct strong interference can also be blocked and when the fraction of reflecting objects are increased the effect of interference through the reflected path increases which causes the coverage to be reduced. In this study we have considered only primary reflections and in analysis we have incorporated only the reflection caused by the nearest visible BSs. It can be observed that strong reflections can also be caused by other nearby reflectors also, which can also cause the coverage to be varied. 

The reflected signal almost covers double the distance of that of direct path most of the cases and hence for the cases with low density of reflectors, the reflectors will not affect the coverage probability, but when the density of the reflector is comparable or more than BS density, interference starts to dominate. This can be seen for the case of $\delta=0.5$.

\section{Conclusion}
In this paper, we have proposed a method to model and analyse the reflections in a mmwave cellular system. We have analysed the coverage probability for cellular networks considering the effects of reflections and blockages. It is found that the coverage probability is sensitive to the presence of objects.  It is noticed that presence of high density of reflectors can improve the coverage in high density networks and in low density networks the reflected signals has to travel longer distances than that of direct path and coverage probability has no further changes from that of blockages. Also it should be mentioned that in this analysis the reflectors are placed randomly and we believe that proper design and placement of reflectors can improve the performance of the network.
 
\bibliographystyle{IEEEtran}
\bibliography{Bib_mmWave}

\begin{thebibliography}{10}
\providecommand{\url}[1]{#1}
\csname url@samestyle\endcsname
\providecommand{\newblock}{\relax}
\providecommand{\bibinfo}[2]{#2}
\providecommand{\BIBentrySTDinterwordspacing}{\spaceskip=0pt\relax}
\providecommand{\BIBentryALTinterwordstretchfactor}{4}
\providecommand{\BIBentryALTinterwordspacing}{\spaceskip=\fontdimen2\font plus
\BIBentryALTinterwordstretchfactor\fontdimen3\font minus
  \fontdimen4\font\relax}
\providecommand{\BIBforeignlanguage}[2]{{%
\expandafter\ifx\csname l@#1\endcsname\relax
\typeout{** WARNING: IEEEtran.bst: No hyphenation pattern has been}%
\typeout{** loaded for the language `#1'. Using the pattern for}%
\typeout{** the default language instead.}%
\else
\language=\csname l@#1\endcsname
\fi
#2}}
\providecommand{\BIBdecl}{\relax}
\BIBdecl

\bibitem{pi2011}
Z.~Pi and F.~Khan, ``An introduction to millimeter-wave mobile broadband
  systems,'' \emph{IEEE Communications Magazine}, vol.~49, no.~6, pp. 101--107,
  June 2011.

\bibitem{rap2013}
T.~S. Rappaport, S.~Sun, R.~Mayzus, H.~Zhao, Y.~Azar, K.~Wang, G.~N. Wong,
  J.~K. Schulz, M.~Samimi, and F.~Gutierrez, ``Millimeter wave mobile
  communications for 5g cellular: It will work!'' \emph{IEEE Access}, vol.~1,
  pp. 335--349, 2013.

\bibitem{rap2016}
G.~R. MacCartney, S.~Deng, and T.~S. Rappaport, ``Indoor office plan
  environment and layout-based mmwave path loss models for 28 ghz and 73 ghz,''
  in \emph{2016 IEEE 83rd Vehicular Technology Conference (VTC Spring)}, May
  2016, pp. 1--6.

\bibitem{al2008}
A.~V. Alejos, M.~G. Sanchez, and I.~Cuinas, ``Measurement and analysis of
  propagation mechanisms at 40 ghz: Viability of site shielding forced by
  obstacles,'' \emph{IEEE Transactions on Vehicular Technology}, vol.~57,
  no.~6, pp. 3369--3380, Nov 2008.

\bibitem{rap2002}
H.~Xu, V.~Kukshya, and T.~S. Rappaport, ``Spatial and temporal characteristics
  of 60-ghz indoor channels,'' \emph{IEEE Journal on Selected Areas in
  Communications}, vol.~20, no.~3, pp. 620--630, Apr 2002.

\bibitem{bend2011}
E.~Ben-Dor, T.~S. Rappaport, Y.~Qiao, and S.~J. Lauffenburger,
  ``Millimeter-wave 60 ghz outdoor and vehicle aoa propagation measurements
  using a broadband channel sounder,'' in \emph{2011 IEEE Global
  Telecommunications Conference - GLOBECOM 2011}, Dec 2011, pp. 1--6.

\bibitem{raj2012}
S.~Rajagopal, S.~Abu-Surra, and M.~Malmirchegini, ``Channel feasibility for
  outdoor non-line-of-sight mmwave mobile communication,'' in \emph{2012 IEEE
  Vehicular Technology Conference (VTC Fall)}, Sept 2012, pp. 1--6.

\bibitem{rap20132}
T.~S. Rappaport, F.~Gutierrez, E.~Ben-Dor, J.~N. Murdock, Y.~Qiao, and J.~I.
  Tamir, ``Broadband millimeter-wave propagation measurements and models using
  adaptive-beam antennas for outdoor urban cellular communications,''
  \emph{IEEE Transactions on Antennas and Propagation}, vol.~61, no.~4, pp.
  1850--1859, April 2013.

\bibitem{ak2014}
M.~R. Akdeniz, Y.~Liu, M.~K. Samimi, S.~Sun, S.~Rangan, T.~S. Rappaport, and
  E.~Erkip, ``Millimeter wave channel modeling and cellular capacity
  evaluation,'' \emph{IEEE Journal on Selected Areas in Communications},
  vol.~32, no.~6, pp. 1164--1179, June 2014.

\bibitem{bai2012}
T.~Bai, R.~Vaze, and R.~W. Heath, ``Using random shape theory to model blockage
  in random cellular networks,'' in \emph{2012 International Conference on
  Signal Processing and Communications (SPCOM)}, July 2012, pp. 1--5.

\bibitem{bai2014}
------, ``Analysis of blockage effects on urban cellular networks,'' \emph{IEEE
  Transactions on Wireless Communications}, vol.~13, no.~9, pp. 5070--5083,
  Sept 2014.

\bibitem{bai2015}
S.~Akoum, O.~E. Ayach, and R.~W. Heath, ``Coverage and capacity in mmwave
  cellular systems,'' in \emph{2012 Conference Record of the Forty Sixth
  Asilomar Conference on Signals, Systems and Computers (ASILOMAR)}, Nov 2012,
  pp. 688--692.

\bibitem{ganti}
J.~G. Andrews, F.~Baccelli, and R.~K. Ganti, ``A tractable approach to coverage
  and rate in cellular networks,'' \emph{IEEE Transactions on Communications},
  vol.~59, no.~11, pp. 3122--3134, November 2011.

\end{thebibliography}

\appendix
  The user in this model connects to both the reflected path and direct path, if they exist. Given the nearest direct BS is at a distance, $r_d$ from the user and the nearest reflecting BS is at a distance $r_r$, we have
\begin{align*}
\P[\sinr>T]
&=\E_{r_d,r_r}\Big[\P[\sinr_D>T|r_d<r_r]\\
&+\P[\sinr_R>T|r_d>r_r]\Big],
\end{align*}
where $\sinr_D = \frac{|h|^2 r_d^{-\alpha}}{\sigma^2+\I_D}$ is the $\sinr$ of user connected to visible BS and $\I_D$ interference seen by the direct connected UE experienced from the reflectors other than the closest reflector will be highly attenuated, we have interference for direct path, \[\I_D=\sum_{\|x\|>r_d,x\in\Phi}\|x\|^{-\alpha}|h_x|^2S_x + r_r^{-\alpha}|h_f|^2S_{r_r}\]
Now
\begin{align}
\lefteqn {\P[\sinr_D>T|r_d,r_r]
= \P\left[|h|^2>Tr_d^\alpha (\I_D+\sigma^2)\right]}\nonumber\\
&\stackrel{(a)}{=}\E\Big[\left(\prod_{\|x\|>r_d,x\in\Phi}e^{-Tr_d^\alpha \|x\|^{-\alpha}|h_x|^2}e^{-\beta \|x\|}\right)+1-e^{-\beta \|x\|} \Big]\nonumber\\
&\times \Big. \E_{h_f}\left(e^{-Tr_d^\alpha r_r^{-\alpha}|h_f|^2 }e^{-\beta r_r}+(1-e^{-\beta r_r})\right) e^{-r_d^{\alpha}{\sigma^2}T}\Big. ,\nonumber\\
&\stackrel{(b)}{=}\E\Big[\left(\prod_{\|x\|>r_d,x\in\Phi}1-\frac{e^{-\beta \|x\|}Tr_d^\alpha\|x\|^{-\alpha}}{1+Tr_d^\alpha\|x\|^{-\alpha}}\right) \Big]\nonumber\\
&\times\left( \frac{1}{1+Tr_d^\alpha r_r^{-\alpha}}e^{-\beta r_r}+(1-e^{-\beta r_r})\right) e^{-r_d^{\alpha}{\sigma^2}T},\nonumber\\
&\stackrel{(c)}{=}\exp\left(-2\pi\lambda\int_{r_d}^{\infty}\left( \frac{Tr_d^\alpha r^{-\alpha}e^{-\beta r}}{1+Tr_d^\alpha r^{-\alpha}}\right)r\d r\right)\nonumber\\
&\times \left(1-\frac{Tr_d^\alpha r_r^{-\alpha}e^{-\beta r_r}}{1+Tr_d^\alpha r_r^{-\alpha}}\right) e^{-r_d^{\alpha}{\sigma^2}T}.
\end{align}
Here (a) by using the fact that $|h|^2\sim \exp(1)$ for Rayleigh fading and (b) by using the PGFL property of PPP.
Now consider the reflected $\sinr_R$, we have,
\[\I_R=\sum_{\|x\|>r_d,x\in\Phi}\|x\|^{-\alpha}|h_x|^2e^{-\beta \|x\|}\] Similar to the above derivation, 
\begin{align}
\lefteqn{\P[\sinr_R>T|r_d,r_r]
= \P\left[|h|^2>Tr_r^\alpha (\I_R+\sigma^2)\right]}\nonumber\\
&= \E\Big[e^{-r_r^{\alpha}{\sigma^2}T}\prod_{\|x\|>r_d,x\in\Phi}\exp\left(-Tr_r^\alpha \|x\|^{-\alpha}|h_x|^2S_x \right)\Big]\nonumber\\
&=e^{-r_r^{\alpha}{\sigma^2}T}\exp\left(-2\pi\lambda\int_{r_d>r_r}^{\infty}\left(\frac{Tr_r^\alpha r^{-\alpha}e^{-\beta r}}{1+Tr_r^\alpha r^{-\alpha}}\right)r\d r\right).
\end{align}

%
%
%
%
%
%
%
%
\end{document}